\documentclass[english, 11pt,a4paper]{article}

\usepackage{microtype}

\pdfoutput=1

\usepackage[left=1in, right=1in, top=1in, bottom=1in]{geometry}
\usepackage{enumitem}
\usepackage{array}
\usepackage{relsize}
\usepackage{titlesec}
\usepackage{amsfonts}
\usepackage{amssymb}
\usepackage{amstext}
\usepackage{amsmath}
\usepackage{amsthm}
\usepackage{mathtools}
\usepackage{microtype}
\usepackage{graphicx}
\usepackage{graphics}
\usepackage{colordvi}
\usepackage{xspace}
\usepackage{footmisc}
\usepackage{paralist}
\usepackage{complexity}
\usepackage[sf,small,bf]{caption}
\usepackage{subcaption}
\usepackage{todonotes}
\newtheorem{theorem}{Theorem}
\newtheorem{lemma}{Lemma}

\usepackage{pifont}

\newclass{\QPTAS}{QPTAS}

\newcommand*\samethanks[1][\value{footnote}]{\footnotemark[#1]}
\newcommand{\ie}{i.\,e.}
\newcommand{\eg}{e.\,g.}
\newcommand{\Cc}{\mathcal{C}}
\newcommand{\Dc}{\mathcal{D}}
\newcommand{\Bc}{\mathcal{B}}
\newcommand{\Rc}{\mathcal{R}}

\newcommand{\paraspace}[1]{\paragraph{#1~~}}

\titlespacing{\paragraph}{0pt}{4pt plus 4pt minus 2pt}{0pt plus 2pt minus 0pt}

\renewcommand{\epsilon}{\varepsilon}

\newcommand\TombStone{\rule{.7ex}{1.7ex}}
\renewcommand{\qedsymbol}{\TombStone}
\newcommand{\qedd}{\let\qed\relax\quad\raisebox{-.1ex}{$\qedsymbol$}}

        {\hspace*{\fill}$\Box$\par\vspace{4mm}}

\title{Maximum Scatter TSP in Doubling Metrics} 

\author{L\'{a}szl\'{o} Kozma\thanks{Department of Computer Science, 
Saarland University,\newline \indent {\tt \  kozma@cs.uni-saarland.de}, \ \ {\tt moemke@cs.uni-saarland.de}}
        \and
        Tobias M\"{o}mke\samethanks~\,\footnote{Funded by Deutsche Forschungsgemeinschaft grant BL511/10-1 and MO 2889/1-1.}
        }

\begin{document}
\maketitle

\begin{abstract}

We study the problem of finding a tour of $n$ points in which \emph{every edge} is \emph{long}. 
More precisely, we wish to find a tour that visits every point exactly once, maximizing the length of the shortest edge in the tour. 
The problem is known as Maximum Scatter TSP, and was introduced by Arkin et al.~(SODA 1997), motivated by applications in manufacturing and medical imaging. 
Arkin et al.~gave a $0.5$-approximation for the metric version of the problem and showed that this is the best possible ratio achievable in polynomial time (assuming $\P \neq \NP$). 
Arkin et al.\ raised the question of whether a better approximation ratio can be obtained in the Euclidean plane. 

We answer this question in the affirmative in a more general setting, by giving 
a $(1-\epsilon)$-approximation algorithm for $d$-dimensional doubling metrics, with running time $\tilde{O}\big(n^3 + 2^{O(K \log K)}\big)$, where $K \leq \left( \frac{13}{\epsilon} \right)^d$.
As a corollary we obtain (i) an efficient polynomial-time approximation scheme ($\EPTAS$) for all constant dimensions $d$, (ii) a polynomial-time approximation scheme ($\PTAS$) for dimension $d = \log\log{n}/c$, for a sufficiently large constant $c$, and (iii) a $\PTAS$ for constant $d$ and $\epsilon = \Omega(1/\log\log{n})$. Furthermore, we show the dependence on $d$ in our approximation scheme to be essentially optimal, unless Satisfiability can be solved in subexponential time. 

\end{abstract}

\section{Introduction}

Let $P = \{p_1, \dots, p_n\}$ be a set of points in some metric space with distance function $d(\cdot)$. A \emph{tour} $T$ of $P$ is a sequence $T = (p_{i_1}, \dots, p_{i_n})$, where $\{i_1,\dots,i_n\} = \{1,\dots,n\}$. The \mbox{\emph{scatter}} of $T$ is the minimum distance between neighboring points, \ie, 
\vspace{-0.08in}
$$\min \left\{d(p_{i_1}, p_{i_2}), \dots, d(p_{i_{n-1}}, p_{i_n}), d(p_{i_n}, p_{i_1})\right\}.$$ \vspace{-0.25in}

The Maximum Scatter Traveling Salesman Problem (MSTSP) asks for a tour of $P$ with maximum scatter. We study this problem in the \emph{geometric} setting where the metric $d(\cdot)$ has specific properties, for instance, when it is the Euclidean distance between points. 

Arkin, Chiang, Mitchell, Skiena, and Yang~\cite{Arkin} initiated the study of MSTSP in 1997, motivated by problems in manufacturing (riveting) and medical imaging. Similarly to other variants of TSP, solving MSTSP exactly is $\NP$-hard. The best we can hope for are \emph{approximations}, \ie, algorithms that compute a tour whose scatter is at least a factor $c$ of the optimum, where $0 < c <1$. We call such an algorithm (as well as its output) a $c$-approximation. We are mostly concerned with algorithms that compute a solution that is at least a factor $1 - \epsilon$ of the optimum, for an arbitrary constant $\epsilon > 0$. Such an algorithm is called a polynomial-time approximation scheme ($\PTAS$), if its running time is polynomial in $n$ for every fixed value of $\epsilon$. An \emph{efficient} polynomial-time approximation scheme ($\EPTAS$) is a $\PTAS$ that has running time $f(\epsilon) \cdot n^k$ for some constant $k$ (not depending on $\epsilon$), and an arbitrary function $f$.

For the metric MSTSP (\ie, the case in which the distance function $d(\cdot)$ is a metric, with no further restrictions), Arkin et al.\ give a $0.5$-approximation. They also show that for this variant of the problem, the ratio of $0.5$ is the best possible in polynomial time (assuming $\P \neq \NP$). They leave open the question of whether a better approximation ratio can be obtained if the problem has geometric structure (\eg, if the distances are Euclidean). Arkin et al.\ raise this question for the two-dimensional case (see also~\cite{TOPP} and \cite[p.\ 681]{Handbook}). 

\paraspace{MSTSP and other TSP variants.} TSP is one of the cornerstones of combinatorial optimization and several variants have been considered in the literature (we refer to~\cite{Gutin, Lawler} for surveys). Minimizing variants are more common, but there exist natural settings in which tours with long edges are desirable. This is the case in certain manufacturing operations where nearby elements in a sequence are required to be geometrically well-separated in order to avoid interferences~\cite{Arkin}. Informally, the difficulty of TSP variants with maximized objective functions lies in the fact that the solutions are highly non-local, making it hard to decompose the problem using standard techniques. (This is true both for algorithms and for hardness-proofs.)

In a certain sense, it is natural to expect that geometric structure should lead to stronger approximation-guarantees for MSTSP. 
The same phenomenon has been observed in case of the standard TSP. 
For metric TSP the best known approximation ratio is $1.5$ (Christofides~\cite{Chr76}), with a current lower bound of $\frac{123}{122}$ (Karpinski et al.~\cite{Karpinski}), whereas for Euclidean TSP, Arora~\cite{arora} and Mitchell~\cite{Mitchell96} independently found a $\PTAS$. 
Similarly, Euclidean MaxTSP (where the goal is to maximize the \emph{total} length of the tour) admits a $\PTAS$ (Barvinok~\cite{barvinok}), but the metric version is known to be $\MaxSNP$-hard~\cite{Yannakakis, barvinok}, ruling out the existence of a $\PTAS$ (assuming $\P \neq \NP$).
On the positive side, metric MaxTSP is currently known to admit a $\frac{7}{8}$-approximation (Kowalik and Mucha~\cite{Kowalik}).

The situation, however, is very different in the case of Bottleneck TSP (a.k.a.\ $\min\max$\ TSP), which is the TSP variant perhaps most similar to MSTSP (a.k.a.\ $\max\min$\ TSP). For metric Bottleneck TSP, a $2$-approximation can be achieved in polynomial time, and this ratio is the best possible, assuming $\P \neq \NP$ (Doroshko and Sarvanov~\cite{Doroshko}, and independently Parker and Rardin~\cite{Parker}). 
Surprisingly, for Bottleneck TSP, geometric structure does \emph{not} help. As shown by Sarvanov\footnote{The result of Sarvanov is available as a technical report in Russian, and appears to be not widely known in the English-language literature (the only reference we are aware of is~\cite{Larusic} and various sources incorrectly mention the problem as still open). We thank Victor Lepin for providing us with a copy of Sarvanov's paper. 
}\cite{Sarvanov} in 1995, the factor of $2$ can not be improved even in the two-dimensional Euclidean case. Based on these results for related problems, it is a priori not clear whether Euclidean MSTSP should admit good approximations.

\paraspace{Results and techniques.}

We answer the open question about Euclidean MSTSP, by giving an efficient polynomial-time approximation scheme ($\EPTAS$) for arbitrary fixed-dimensional Euclidean spaces. 
Since MSTSP is known to be strongly $\NP$-hard in dimensions $3$ and above~\cite{Fekete, BFJ}, our result settles the classical complexity status of the problem in these dimensions. 

If the dimension $d$ is superconstant (\ie, depending on $n$), we still obtain a $\PTAS$ if the dependence on $n$ is mild, \ie, $d = \log\log{n}/c$ for a sufficiently large constant $c$. 
Alternatively, for constant $d$ we may choose $\epsilon$ to be asymptotically smaller than 1 depending on $n$, \eg, as $\epsilon = \Theta(1/\log\log{n})$, and still obtain a $\PTAS$.

Our approximation scheme carries over to the more general case of \emph{doubling metrics}. The \emph{doubling dimension} of a metric space $\Dc$ with metric $d(\cdot)$ is the smallest $k$ such that every ball of radius $r$ in $\Dc$ can be covered by $2^k$ balls of radius $r/2$. If $k$ is finite, we also refer to $\Dc$ and $d(\cdot)$ as a \emph{doubling space}, resp.\ \emph{doubling metric}. Metric spaces with doubling dimension $O(d)$ generalize the Euclidean space with $d$ dimensions, and are significantly more general (see \eg, \cite{Gupta, Bartal, Clarkson1999, Talwar} and references therein). 

Arguments for Euclidean spaces do not always extend to doubling spaces. For instance, in the case of the standard TSP, the question of whether a $\PTAS$ exists in doubling spaces has been open for many years~\cite{Talwar, Bartal}, and was settled only recently by Bartal, Gottlieb and Krauthgamer~\cite{Bartal}. 

Our main result is the following. 
\begin{theorem}[$\EPTAS$]\label{thm2}
Let $P$ be a set of $n$ points in a metric space with doubling dimension $d$. A tour of $P$ whose scatter is at least a $(1-\epsilon)$ factor of the MSTSP optimum can be found in time $O(n^3 \log{n} + 2^{O(K \log K)}\log{n} + K^3 \log^2{n})$, where $K \leq \left( \frac{13}{\epsilon} \right)^d$. 
\end{theorem}

We show that the running time in Theorem~\ref{thm2} is essentially tight (apart from the involved constants and lower order factors). A result of Trevisan~\cite{Trevisan} shows that for $\Omega(\log{n})$-dimensional Euclidean spaces, TSP is $\APX$-hard, and thus admits no $\PTAS$, unless $\P = \NP$. Adapting the proof technique of Trevisan to MSTSP, we show that an approximation scheme for MSTSP with an improved dependence on $d$ would give a subexponential-time algorithm for TSP in bipartite cubic graphs, contradicting the exponential-time hypothesis (ETH). More precisely, we show the following. 

\begin{theorem}\label{thm:eth}
There is a constant $c$ such that a $(3/4)^{1/p}$-approximation algorithm with running time $2^{2^{o(d)}}$ for MSTSP in $\mathbb{R}^d$ for $d \ge c\log n$ and distances according to the $\ell_p$ norm contradicts the ETH.
\end{theorem}

The main ingredient of the algorithm is a rounding scheme, whereby we replace the original point set with points of a grid (or an $\epsilon$-net). 
As subproblems, we solve perfect matching and Eulerian tour instances, and we use the many-visits TSP algorithm of Cosmadakis and Papadimitriou~\cite{Cosmadakis}. 

In order to introduce some of the ideas and tools used for proving the result, we first describe a $\PTAS$ that is simpler than our main algorithm but has weaker guarantees. This result already settles the open question about MSTSP by using only elementary techniques. 
However, due to its prohibitive running time, it is only of theoretical (and perhaps didactic) interest. 

The easier $\PTAS$ is described in \textsection\,\ref{sec1}. The $\EPTAS$ (Theorem~\ref{thm2}) is given in \textsection\,\ref{sec2}. The hardness proof can be found in \textsection\,\ref{sec4}.

\paraspace{Further related work.}
Searching for approximation schemes for $\NP$-hard geometric optimization problems (or showing that they cannot exist) is a rich field of research. Such schemes, while often impractical to implement directly, have led both to new algorithmic techniques and to a better structural understanding of the geometric problems. A full survey of such works is out of scope here, we only mention a few representative or recent results~\cite{arora, Mitchell96, Friggstad, sunil-arya, Saurabh, sariel_book}.

\paraspace{Open question.}
Our current work does not address the complexity status of solving MSTSP in the Euclidean plane \emph{exactly}. It remains open whether this problem is $\NP$-hard (the situation is the same for MaxTSP). We note that this question has a natural equivalent formulation: Is checking for existence of a Hamiltonian cycle $\NP$-complete in complements of unit disk graphs?

\section{Warmup: a simple $\PTAS$}\label{sec1}

In this section we prove the following theorem. To keep the analysis simple, we do not optimize the involved constants, and we restrict attention to Euclidean spaces.

\begin{theorem}
Let $P$ be a set of $n$ points in $\mathbb{R}^d$. A tour of $P$ whose scatter is at least a $(1-\epsilon)$ factor of the MSTSP optimum can be found in time $O\left(n^{(75d/\epsilon^2)^d} \right)$.
\end{theorem}

Consider a set $P$ of $n$ points in $\mathbb{R}^d$, and a precision parameter $\epsilon>0$. Observe that the MSTSP optimum (\ie, the shortest edge length of a tour of $P$) can only take one of ${n \choose 2}$ possible values (the distances between points in $P$). We sort the distances in decreasing order, and for each distance $\ell$ we attempt to construct a tour with scatter $\ell$. Our procedure will either produce a certificate that no such tour exists (in which case we continue with the next value of $\ell$), or it succeeds in constructing a tour of $P$ with scatter at least $\ell(1-\epsilon)$.

Before proceeding to the algorithm, we present a simple structural observation upon which the algorithm relies. 
Suppose that $\ell$ is the maximum value such that a tour $T$ of $P$ exists with scatter $\ell$, \ie, $\ell$ is the MSTSP optimum. Let $p,q \in P$ be two neighboring points in $T$, \ie, the edge $\{p,q\}$ is part of $T$, such that $d(p,q) = \ell$. 
Let $B_p$ and $B_q$ denote the open balls of radius $\ell (1 + \epsilon)$, with centers $p$ and $q$, respectively. We show that the optimal solution can be assumed to have a certain structure in relation to $B_p$ and $B_q$.

\begin{lemma}
\label{lem3a}
Let $\ell$ be the maximum value such that a tour $T$ of $P$ exists with scatter $\ell$. Then there is a tour $T'$ of $P$ with scatter $\ell$ containing an edge $\{p,q\}$ with $d(p,q) = \ell$, such that for every edge $\{x,y\}$ of $T'$, at least one of $x$ and $y$ is contained in $B_p \cup B_q$.
\end{lemma}

\begin{proof}
Let $T$ be a tour of $P$ with scatter $\ell$, let $\{p,q\}$ be an edge of $T$ of length $\ell$, and let $\{x,y\}$ be an edge of $T$ such that $x,y \notin B_p \cup B_q$. By the definition of $B_p$ and $B_q$ we have $d(x,p), d(x,q), d(y,p), d(y,q) > \ell$. Thus, we can replace the edges $\{x,y\}$ and $\{p,q\}$ in $T$, with either $\{x,p\}$ and $\{y,q\}$, or $\{x,q\}$ and $\{y,p\}$, depending on the ordering of the points in $T$. We obtain another tour with scatter at least $\ell$, in which the number of edges of length exactly $\ell$ has decreased. If the obtained tour contains no edge of length $\ell$, then its scatter is strictly larger than $\ell$, contradicting the choice of $\ell$. Otherwise, we repeat the argument with a remaining edge of length $\ell$. See Fig.~\ref{fig1a} for an illustration. \qedd
\end{proof}

\begin{figure}[tb]
\centering
\includegraphics[width=2in]{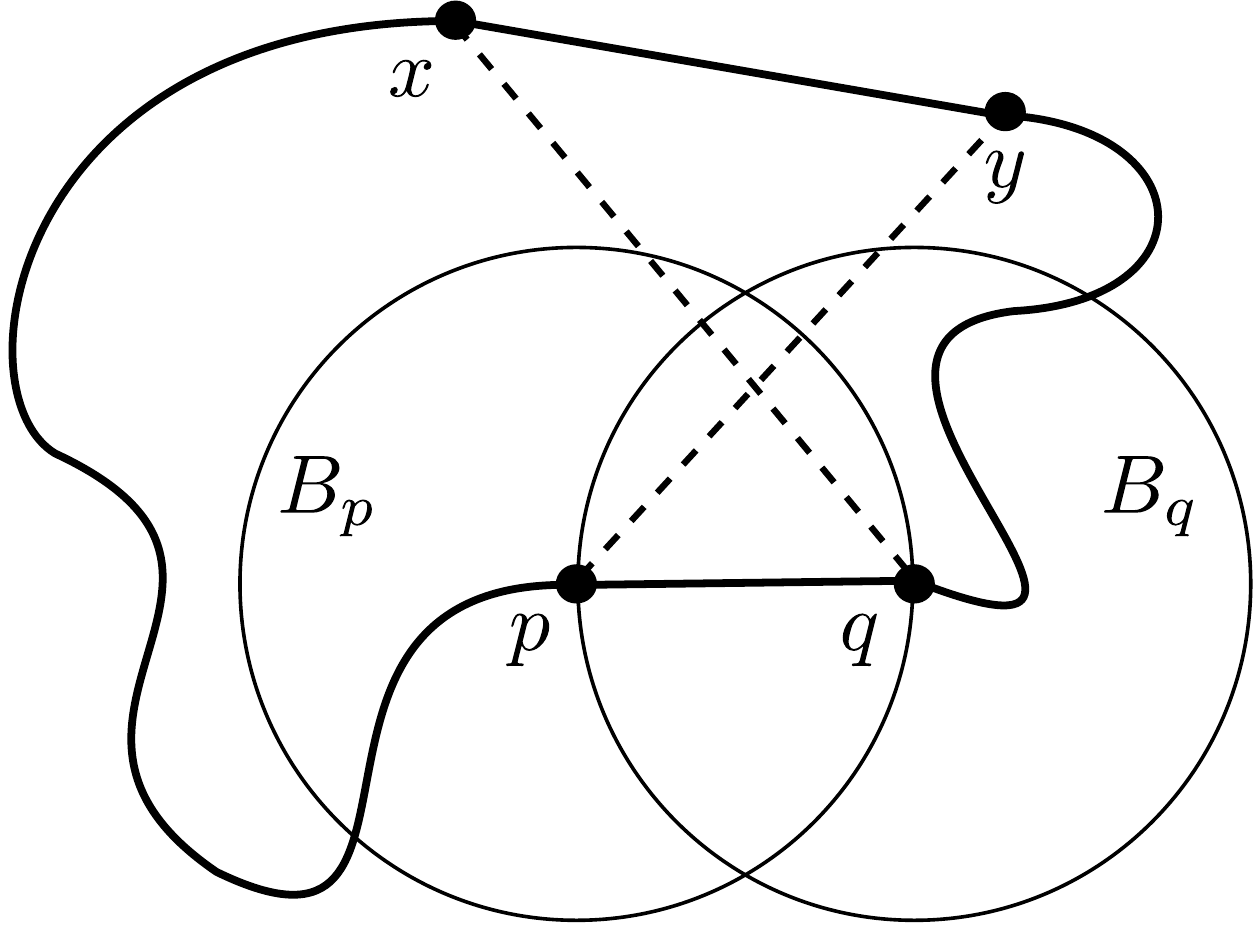}
\caption{Illustration of Lemma~\ref{lem3a}. The dashed edges can replace $\{x,y\}$ and $\{p,q\}$ in the optimal tour.}
\label{fig1a}
\end{figure}

The next ingredient of the algorithm is a coarsening of the input, by rounding points in $P$ to points of a grid. Let $\mathbb{G}_{\delta}$ be a $\delta$-scaling of the $d$-dimensional unit grid, \ie, $\mathbb{G}_{\delta} = \{ \delta (n_1, \dots, n_d) \mid n_1, \dots, n_d \in \mathbb{Z}\}$, for an arbitrary $\delta > 0$. Let $f_\delta(\cdot)$, or simply $f(\cdot)$, be the function from $\mathbb{R}^d$ to $\mathbb{G}_\delta$ that maps each point to its nearest grid point (breaking ties arbitrarily). The following properties result from basic geometric considerations.

\begin{lemma}
\label{lem4}
With $f(\cdot)$ and $\delta$ as defined earlier, we have:
\begin{compactenum}[(i)]
\item $d(x,y) \geq d(f(x),y) - \delta \sqrt{d}/2$ for all $x,y \in  \mathbb{R}^d$, 
\item $|B \cap \mathbb{G}_{\delta}| \leq (2 \ell / \delta + 1)^d$ for every open ball $B$ of radius $\ell$. 
\end{compactenum}
\end{lemma}

Given a point set $P \in \mathbb{R}^d$, let $G_P$ be a graph with vertex set $V(G_P) = P$ and edge set $E(G_P) = \left\{ \{x,y\} \mid x,y\in P \ \land \ d(x,y) \geq \ell \right\}$. In words, $G_P$ contains all edges with length at least $\ell$. We wish to find a Hamiltonian cycle in $G_P$ or show that there is none.

Observe that under $f(\cdot)$, the graph $G_P$ becomes a multigraph $H_P$ defined as follows: $V(H_P) = \{v \mid v = f(x), x \in P\}$, \ie, the set of grid points with at least one mapped point of $P$, and $E(H_P) = \left\{\{u,v\} \mid u = f(x), v = f(y), \{x,y\} \in E(G_P)\right\}$, \ie, the pairs of grid points to which edges of $G_P$ are mapped. We also maintain multiplicities on edges of $H_P$, \ie, we keep track of how many edges of $G_P$ are mapped to each edge of $H_P$.

Consider a tour $T$ of $P$ of scatter at least $\ell$, \ie, a Hamiltonian cycle of $G_P$. Under $f(\cdot)$, the edges of $T$ form a spanning subgraph $H'_P$ of $H_P$ (\ie, $V(H'_P) = V(H_P)$), and $T$ is mapped to an \emph{Eulerian tour} of $H'_P$.  
It is not hard to see that given this Eulerian tour of $H'_P$, a tour of $P$ can be recovered (by replacing multiple occurrences of a grid point with the points in $P$ that are mapped to it). Moreover, the scatter of the recovered tour is not far from that of $T$ (by Lemma~\ref{lem4}(i)). However, the edges of $H'_P$ are not available to us, and guessing them seems prohibitively expensive.  
The key idea is that it is sufficient to consider the portion of $H_P$ that falls inside $B_p \cup B_q$, for points $p$ and $q$ chosen as in Lemma~\ref{lem3a}. \\

We describe at a high level the approximation algorithm (see algorithm~A for the details, and Fig.~\ref{fig2} for an illustration). We assume that the optimal tour $T$ has the property from Lemma~\ref{lem3a} with respect to points $p,q \in P$, and that such $p$ and $q$ can be ``guessed''. Let $\Bc = B_p \cup B_q$ (recall that $B_p$ and $B_q$ have radius $\ell (1 + \epsilon)$.) Then, $T$ consists of edges fully inside $\Bc$, and ``hops'' of two consecutive edges, connecting a point outside $\Bc$ with two points inside $\Bc$. 
We replace such hops with virtual edges, both of whose endpoints are in $\Bc$. The resulting tour is entirely in $\Bc$, and we can ``guess'' its image under $f_{\delta}(\cdot)$. This is now feasible, since the number of grid points involved is bounded by Lemma~\ref{lem4}(ii). We also guess the multiplicities of all edges, \ie, how many original edges have been mapped to each edge, and how many edges are virtual. 

We then disambiguate the virtual edges, finding a suitable midpoint outside of $\Bc$ for each hop. This is achieved by solving a perfect matching problem. We obtain a multigraph on which we find an Eulerian tour. Finally, from the Eulerian tour we recover a tour of $P$. The distortion in distances due to rounding (\ie, the approximation ratio) is controlled by the choice of the grid resolution $\delta$.

\paraspace{Algorithm A ($\PTAS$).}\ \\ {\sc Input:} Set $P$ of $n$ points in $\mathbb{R}^d$ and a precision parameter $\epsilon > 0$.\\
{\sc Output:} A value $\ell$ such that $P$ admits a tour with scatter at least $\ell(1 - \epsilon)$, and no tour with scatter greater than $\ell$. \\

\vspace{-0.12in}
\begin{compactenum}
\item {\bf Iterate $p,q$ over all pairs of points in $P$ in decreasing order of $d(p,q)$.}

\item Let $\ell = d(p,q)$. Let $\delta = \epsilon \ell / (2 \sqrt{d})$. Let $\ell' = \ell(1 - \epsilon/2)$.

Let $B_p$ and $B_q$ be the open balls with centers $p$ and $q$ of radius $\ell (1 + \epsilon)$, and let $\Bc = B_p \cup B_q$. 

\item Let $f\colon P \rightarrow \mathbb{G}_\delta$ map points to their nearest grid point. 

Compute the set $\Cc = \{ f(x) \mid x \in (P \cap \Bc)\}$, and for each $v\in \Cc$, compute the sets $f^{-1}(v) = \{x \mid f(x) = v\}$.
\item Let $m,v: {\Cc \choose 2} \rightarrow \mathbb{N}$. For all $\{u,v\} \subseteq \Cc$, \emph{\bf guess} $m(\{u,v\})$ and $v(\{u,v\})$, such that\\ (i) $m(\{u,v\}) = 0$ \ if \  $d(u,v) < \ell'$, and\\
(ii) for all $v\in \Cc$:\\
$\displaystyle\sum_{u \in \Cc \setminus \{v\}}{\!\!\!\big( m(\{u,v\}) + v(\{u,v\}) \big)} = 2  |f^{-1}(v)|$. 
 \label{step:guess}

\item Construct a bipartite graph $B$ as follows:\\ 
- for each $\{u,v\} \subseteq \Cc$, add $v(\{u,v\})$ vertices labeled $\{u,v\}$ to the left vertex set $L(B)$.\\
- for each $x \in P \setminus \Bc$ add a vertex labeled $x$ to the right vertex set $R(B)$.\\
- add edge $(\{u,v\},x)$ between $\{u,v\} \in L(B)$ and $x \in R(B)$ to $E(B)$ iff ${d(u,x), d(v,x) \geq \ell'}$.

\item Find a perfect matching $M$ of $B$; if there is none, output {\sc No}.
\item Construct a multigraph $H$ as follows:\\
- let $V(H) = \Cc \cup (P \setminus \Bc)$.\\
- for all $\{u,v\} \subseteq \Cc$ add $m(\{u,v\})$ copies of the edge $\{u,v\}$ to $E(H)$.\\
- for all $(\{u,v\},x) \in M$ add the edges $\{u,x\}$ and $\{v,x\}$ to $E(H)$.
\item Find an Eulerian tour $Q$ of $H$; if there is none, output {\sc No}.
\item Transform $Q$ into a tour $T$ of $P$, by replacing multiple occurrences of every point $v \in \Cc$ with the points in $f^{-1}(v)$ in arbitrary order. 

\item Output {\sc Yes}.\\
\end{compactenum}

\vspace{-0.04in}
\emph{Note}. The ``guessing'' in step~\ref{step:guess} should be thought of as a loop over all possible values of $m(\cdot)$ and $v(\cdot)$ satisfying the requirements. An output of {\sc No} in step~6 or in step~8 means that we proceed to the next values in the loop of step~4, or, if all values have been tried, then we proceed to the next iteration of the main loop in step~1.

\paraspace{Correctness.}

We prove two claims which together imply that algorithm~A is a $\PTAS$ for MSTSP: (1) if algorithm~A outputs {\sc Yes}, then there is a tour of $P$ with scatter at least $\ell(1-\epsilon)$, and (2) if there is a tour of $P$ with scatter at least $\ell$, then algorithm~A outputs {\sc Yes}.

(1) Suppose that algorithm~A returns {\sc Yes} in step 10, and thus steps 5--9 were successful with the current choice of $p,q \in P$ and the values of $m(\cdot)$ and $v(\cdot)$ chosen in step 4, and $T$ is a tour of $P$. Consider an arbitrary edge $\{x,y\}$ of $T$. By step 9, there is a corresponding edge $\{u,v\}$ in the Eulerian tour $Q$ of $H$. By construction of $H$ in step 7, either (a) $u,v \in \Cc$, or (b) $(\{u,w\},v) \in M$ or $(\{v,w\},u) \in M$, for some grid point $w \in \Cc$. 

In case (a) by condition (i) of step 4, we have $d(u,v) \geq \ell'$. Since $\{u,v\} = \{f(x), f(y)\}$, from Lemma~\ref{lem4}(i) we obtain $d(x,y) \geq \ell' - \delta \sqrt{d} = \ell(1 - \epsilon)$.

In case (b) by the construction of $B$ in step 5, we have $d(u,v) \geq \ell'$. Since $\{u,v\}$ equals either $\{f(x), y\}$ or $\{x, f(y)\}$, from Lemma~\ref{lem4}(i) we obtain $d(x,y) \geq \ell' - \delta \sqrt{d}/2 \geq \ell(1 - \epsilon)$.\\

\vspace{-0.1in}
(2) Assume now that a tour $T$ of $P$ with scatter at least $\ell$ exists. Assume also w.\,l.\,o.\,g.\ that $T$ has the special structure described in Lemma~\ref{lem3a} with respect to $p$ and $q$, \ie, it consists of hops and of edges entirely inside $\Bc = B_p \cup B_q$. 
Consider an edge $\{x,y\}$ of $T$, such that $x,y \in \Bc$. Then, after step 3, $f(x),f(y) \in \Cc$ holds, and we say that $\{x,y\}$ \emph{maps} to $\{f(x), f(y)\}$. 
Consider now a hop of $T$, \ie, two consecutive edges $\{x,w\}$ and $\{w,y\}$, such that $x,y \in \Bc$ and $w \in P \setminus \Bc$. Then, after step 3, $f(x),f(y) \in \Cc$ holds, and we say that the hop $\{x,w,y\}$ \emph{virtually maps} to $\{f(x), f(y)\}$.

Consider now the values $m(\cdot)$ and $v(\cdot)$ guessed in step 4, and let $m^*(\{u,v\})$ be the number of edges in $T$ that map to $\{u,v\}$, and let $v^*(\{u,v\})$ be the number of hops in $T$ that virtually map to $\{u,v\}$. Since every point in $T$ has degree $2$, it follows that the number of edges and hops mapped to an edge incident to some $u \in \Cc$ is twice the number of points in $P$ mapped to $u$. 
Furthermore, for all edges $\{x,y\} \subseteq \Bc$ of $T$, we have $d(f(x),f(y)) \geq \ell - \delta \sqrt{d} = \ell'$ (by Lemma~\ref{lem4}(i)). Therefore, guessing the correct values $m=m^*$ and $v=v^*$ is consistent with the conditions in step 4.

Let $\{x_1, w_1, y_1\}, \dots, \{x_k, w_k, y_k\}$ denote all the hops in $T$, where $w_i \in P \setminus \Bc$, for all $i$. Let $u_i = f(x_i)$, and $v_i = f(y_i)$, and let us call $M(T) = \big\{ \big(\{u_i,v_i\}, w_i\big) \mid i=1,\dots,k\big\}$ the \emph{hop-matching} of $T$. Observe that $M(T)$ is a valid perfect matching for the graph $B$ constructed in step 5, therefore, step 6 will succeed. We cannot, however, guarantee that $M(T)$ will be recovered in step 6. Observe that any other perfect matching $M$ of $B$ corresponds to a shuffling of the points $w_i$ in $B$, and thus it is a hop-matching of a tour $T'$ in which the points $w_i$ have been correspondingly shuffled. $T'$ differs from $T$ only in its hops, and by construction of $B$ in step 5, we see that $T'$ has a scatter at least $\ell' - \delta\sqrt{d}/2 \geq \ell(1 - \epsilon)$. 

It can be seen easily that the edges of $T'$ are mapped to an Eulerian tour of the multigraph $H$ constructed in step~7, and thus, step~8 succeeds. Again, we cannot guarantee that the recovered Eulerian tour is the same as the one to which $T'$ maps. Every Eulerian tour of $H$, however, has to respect the edge-multiplicities of $H$, which in turn are determined by the number of points that map to each vertex of $H$. Therefore, step~9 succeeds, and the output is {\sc Yes}.

\paraspace{Running time.}

The loop of step~1 represents a factor of $O(n^2)$ in the running time.
The cost of steps 2--3 is dominated by the cost of the loop that starts in step 4.

We need to bound $|\Cc|$, \ie, the number of grid points inside $\Bc$, to which some point of $P$ is mapped. By Lemma~\ref{lem4}(i) all such grid points are inside one of the balls with center $p$ and $q$, with radius $\ell + \epsilon + \delta \sqrt{d} / 2$. 

\pagebreak

By Lemma~\ref{lem4}(ii) we have 
$$|\Cc| \leq 2 \cdot \left(\frac{2\ell + 2 \epsilon \ell + \delta \sqrt{d}}{\delta} + 1\right)^d = 2 \cdot (4\sqrt{d}/\epsilon + 5\sqrt{d} + 1)^d \leq 2 \cdot (5\sqrt{d}/\epsilon)^d.$$

(In the last inequality we assume $\epsilon \leq 1/5$.)

Steps 5 and 6 amount to finding a perfect matching, and steps 7 and 8 amount to finding an Eulerian tour, both in a graph with $O(n)$ vertices. These steps can be executed in time $O(n^3)$ using standard algorithms. 
As for step 4, observe that the values of $m(\cdot)$ and $v(\cdot)$ over all pairs in $C$ sum to $|P \cap \Bc| \leq n$, so we need to consider at most ${n \choose |\Cc|^2}$ ways of distributing a value of at most $n$ into the integer values of $m$ and $v$. 
Multiplying, and using a standard bound on the binomial coefficient, we obtain that the running time is at most $O(n^{|\Cc|^2 + 3 + 2}) \preceq O\left(n^{(75d/\epsilon^2)^d}\right)$. \\

\vspace{-0.1in}
\emph{Note}. The restrictions on $m(\cdot)$ and $v(\cdot)$ in step \ref{step:guess} can be strengthened, resulting in a smaller number of iterations (and thus better running time). For instance, since each virtual edge corresponds to a hop via a point outside of $\Bc$, we could require the values of $v(\cdot)$ to sum to $|P \setminus \Bc|$. 
 We ignore such technicalities, as they do not affect the correctness of Algorithm~A -- in the case of wrong values, we get the {\sc No} output in some of the later steps. 

\begin{figure*}[t]
\center
\includegraphics[width=0.95\textwidth]{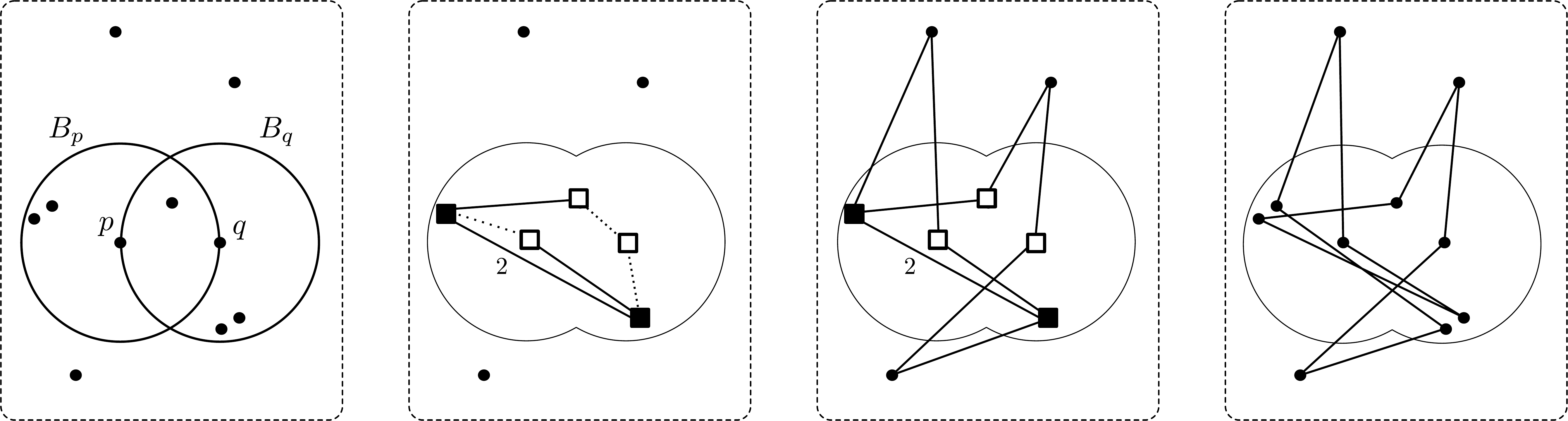}
\caption{Illustration of algorithm~A. (i) Input point set with open balls $B_p$ and $B_q$ with centers $p$ and $q$ and radius $\ell(1+\epsilon)$. (ii) Points inside $B_p \cup B_q$ mapped to grid points (shown as squares), and ``guessed'' edges. Filled squares indicate grid points to which more than one point is mapped. Dotted lines indicate virtual edges, and the number indicates the multiplicity of an edge (omitted if 1). (iii) Virtual edges matched to points outside of $B_p \cup B_q$ and extended to hops, resulting in a multigraph. (iv) An Eulerian tour of the multigraph, expanded to a tour on the initial point set.}
\label{fig2}
\end{figure*}

\section{The main result: an $\EPTAS$}\label{sec2}

In this section we strengthen the result of \textsection\,\ref{sec1} and prove Theorem~\ref{thm2}.

Let $\Dc$ be a doubling space with metric $d(\cdot)$, and let $d$ denote the doubling dimension of $\Dc$. Let $P$ be a set of $n$ points in $\Dc$. 
Suppose we are given a threshold value $\ell>0$, together with a precision parameter $\epsilon > 0$. Given these inputs, we seek an approximation for the \emph{decision version} of MSTSP: With running time polynomial in $n$, our algorithm returns ``yes'' if a tour of $P$ exists with scatter at least $\ell$, and returns ``no'' if there is no tour of $P$ with scatter at least $\ell(1-\epsilon)$. If neither of the two conditions hold, the algorithm is allowed to return ``yes'' or ``no'' arbitrarily. 
Since the optimum value of an MSTSP instance can only take one of the ${n \choose 2}$ possible distances, we can turn an approximation of the above kind into an approximation for the optimization problem with a simple binary search, at the cost of a $(\log{n})$-factor in the running time. In the following, we focus only on the decision problem. 

We start by presenting the main tools that we use in our algorithm. 
  
The following two graph-theoretic results are well-known (see \eg, \cite{bollobas}). Dirac's theorem (Lemma~\ref{dirac}) was also used by Arkin et al.\ to get a $0.5$-approximation for metric MSTSP. The Bondy-Chv\'atal theorem (Lemma~\ref{bondy}, \cite{BondyChvatal}) is a generalization of Lemma~\ref{dirac}. 
The classical proofs of both results are implicitly algorithmic.

\begin{lemma}[Dirac's theorem]\label{dirac}
A graph $G$ with $n$ vertices has a Hamiltonian cycle if the degree of every vertex in $G$ is at least $\frac{n}{2}$. Furthermore, in such a case, a Hamiltonian cycle can be found in $O(n^2)$ time.
\end{lemma}

\begin{lemma}[Bondy-Chv\'atal theorem]\label{bondy}
Let $G = G_0$ be a graph with $n$ vertices, and let $G_1, \dots, G_k$ be a sequence of graphs. For all $i$, the graph $G_i$ is obtained from $G_{i-1}$ by adding an edge between some pair of vertices whose total degree in $G_{i-1}$ is at least $n$. The graph $G_k$ is Hamiltonian if and only if $G$ is Hamiltonian. Furthermore, given $G_k$ and a Hamiltonian cycle of $G_k$, a Hamiltonian cycle of $G$ can be found in $O(n^3)$ time.
\end{lemma}

Let $G_P$ denote the graph over $P$ containing all edges with length at least $\ell$, similarly as in \textsection\,\ref{sec2}. Observe that if the condition of Dirac's theorem (Lemma~\ref{dirac}) holds for $G_P$, then we are done. In the following, we assume that this is not the case. Hence, there is a vertex in $G_P$ whose degree is less than $\frac{n}{2}$. In other words, there is a point $p \in P$, such that $|B_p \cap P| > \frac{n}{2}$, where $B_p$ is the open ball of radius $\ell$ with center $p$. Let us fix $p$ to be such a point, and let $B_p$ and $B'_p$ denote the open balls of radius $\ell$, resp.\ $2 \ell$ with center $p$. We make a simple structural observation similar to Lemma~\ref{lem3a}.

\begin{lemma}
\label{lem3}
Suppose a tour $T$ of $P$ with scatter at least $\ell$ exists. Then there exists a tour $T'$ of $P$ with scatter at least $\ell$, such that for every pair $x,y \in P$ of neighboring points in $T'$, at least one of $x$ and $y$ is contained in $B'_p$.
\end{lemma}

\begin{proof}
Suppose this is not the case. Since $B_p$ contains more than half of the points in $P$, it has to contain at least one edge of $T$ entirely. Let $\{z,t\}$ be such an edge. Since both $x$ and $y$ are outside of $B'_p$ we have $d(x,t), d(x,z), d(y,t), d(y,z) \geq \ell$. Thus, we can replace the edges $\{x,y\}$ and $\{z,t\}$ in $T$, with either $\{x,z\}$ and $\{y,t\}$, or $\{x,t\}$ and $\{y,z\}$, depending on the ordering of the points in $T$. We obtain another tour with scatter at least $\ell$, that no longer contains the edge $\{x,y\}$. We proceed in the same way until we have removed all edges with both endpoints outside of $B'_p$. See Fig.~\ref{fig1} for an illustration. \qedd
\end{proof}

\begin{figure}[tb]
\centering
\includegraphics[width=2in]{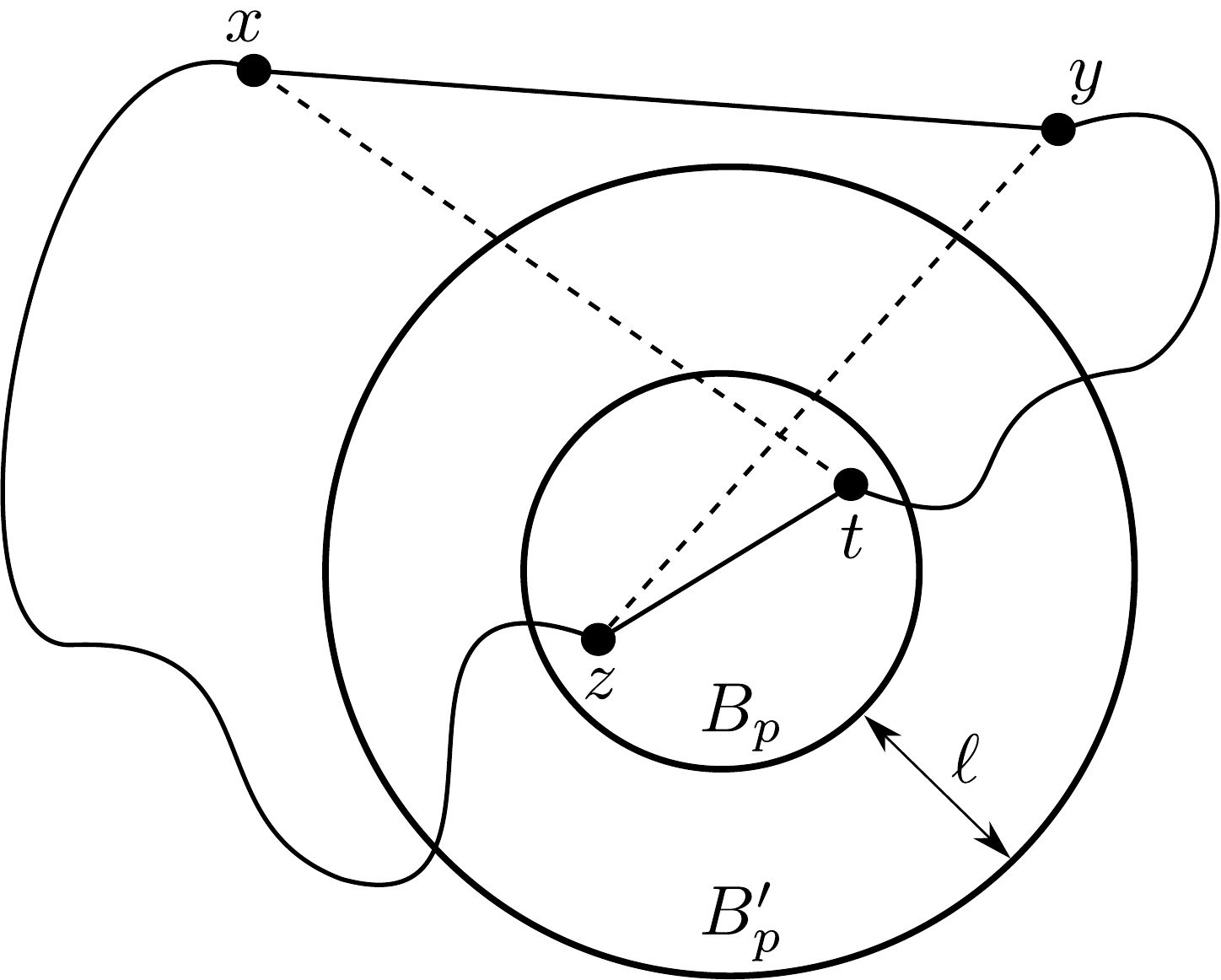}
\caption{Illustration of Lemma~\ref{lem3}. The dashed edges can replace $\{x,y\}$ and $\{z,t\}$ in the optimal tour.}
\label{fig1}
\end{figure}

The next tool we use is an efficient algorithm for the \emph{many-visits} variant of TSP, found in 1984 by Cosmadakis and Papadimitriou~\cite{Cosmadakis}. We state the result in the following variant.  

\begin{lemma}[Many-visits TSP]\label{lemcp}
Given a graph $G$ with vertices $v_1, \dots, v_k$, and integers $n_1, \dots, n_k > 0$, there is an algorithm CP that finds a tour of $G$ that visits vertex $v_i$ exactly $n_i$ times, for all $i$, or reports that no such tour exists.

The running time of algorithm CP is $O \big( n + k^3 \cdot \log{n} + 2^{O(k \log {k})} \big)$, where $n = \sum{n_i}$.
\end{lemma}

Lemma~\ref{lemcp} is shown by Cosmadakis and Papadimitriou in the variant where $G$ is the complete graph, and for arbitrary given edge-weights the shortest tour is to be found (respecting the given multiplicities $n_i$). To obtain the statement in Lemma~\ref{lemcp} it is sufficient to choose suitable edge-weights, \eg, the values $1$ and $2$ for wanted, resp.\ unwanted edges. We then need to check whether the length of the shortest tour is at most $\sum{n_i}$, \ie, whether a tour exists that uses only weight-one edges. 

In its original form, algorithm CP returns only the multiplicities of edges in the solution, not the tour itself. We can find a tour with the obtained edge-multiplicities by computing an Eulerian tour, with running time $O\big(\sum{n_i} \big)$. We refer to~\cite{Cosmadakis} for more precise bounds on the running time. 

Finally, we replace the grid-cover used in \textsection\,\ref{sec1} by an $\epsilon$-net with stronger properties. 

Consider a set of points in $\Dc$. We call a point set $\Rc_P$ in $\Dc$ an \emph{$\epsilon$-net} of $P$, if the balls of radius $\epsilon$ centered at the points in $\Rc_P$ cover $P$. Given an $\epsilon$-net $\Rc_P$ of $P$, we define a function $f(\cdot)$ that maps every point in $P$ to the nearest point in $\Rc_P$. We state the following.

\begin{lemma}\label{lem6}
Given an open ball $B$ of radius $\ell$ in in $\Dc$, $\delta > 0$, and a point set $P \subset B$, we can find a $\delta$-net $\Rc$ of $P$, in deterministic $O(|P|^2)$ time, with the following properties.
\begin{compactenum}[(i)]
\item $d(x,y) \geq d(f(x),y) - \delta$ for all $x,y \in \Rc$.
\item $|\Rc| \leq \big( 2 \ell  / \delta + 1 \big) ^d$.
\end{compactenum}
\end{lemma}

\begin{proof}
Part (i) follows from the definition of a $\delta$-net.
We construct $\Rc$ greedily by repeatedly picking an unmarked point in $P$, adding it to $\Rc$, and marking it, together with all points in $P$ that are at distance at most $\delta$ from the picked point, until all points in $P$ are marked. We observe that the balls of radius $\delta/2$ centered at the points in $\Rc$ are disjoint and contained in a ball of radius $\ell + \delta/2$. By the fact that the size of a packing is at most the size of the minimum cover, and observing that $\big( 2 \ell  / \delta + 1 \big) ^d$ balls of radius $\delta/2$ can cover an arbitrary ball of radius  $\ell + \delta/2$, we obtain the bound of part (ii). In the last step we repeatedly apply the bound from the definition of a metric space with doubling dimension $d$, and we assume for simplicity that $2\ell/\delta + 1$ is a power of two. 
\qedd
\end{proof}

We are ready to present the algorithm. Consider a set $P$ of points in $\Dc$, a threshold value $\ell$, and a precision parameter $\epsilon > 0$. We solve the approximate decision problem described above. Recall that $G_P$ is the graph over $P$, containing the edges with length at least $\ell$.

Call a point $p \in P$ \emph{low-degree}, if $|B_p \cap P| > \frac{n}{2}$, where $B_p = B_1$ is the open ball of radius $\ell$ with center $p$. If no such point exists, we are done, by Lemma~\ref{dirac}. Let $B_2$ be an open ball of radius $2\ell$ with a low-degree point $p$ at its center. Observe that all low-degree points are contained in $B_2$. Suppose otherwise that there is a low-degree point $q$ outside of $B_2$. Then the balls $B_p$ and $B_q$ are disjoint and contain more than $n$ points in total, a contradiction. (A stronger bound on the radius of $B_2$ can be shown, as a function of $d$, but the value of $2\ell$ is sufficient for our purposes.) Let $B_3$ denote the open ball of radius $3 \ell$ having the same center as $B_2$. 

Algorithm~B uses a similar rounding approach as algorithm~A. We first sketch the main idea, then present the details of the algorithm.\\

We wish to find a tour of $P$ (\ie, a Hamiltonian cycle of $G_P$). 
Given $B_2$ and $B_3$ as described above, and using Lemma~\ref{bondy}, we can augment $G_P$ with all edges among points outside $B_2$ (since these are all high-degree points). Furthermore, observe that every edge between a point outside $B_3$ and a point inside $B_2$ is already present (since the distance between such points is at least $\ell$). Thus, every point outside of $B_3$ can be assumed to be connected to every point in $P$. We can therefore collapse every point in $P \setminus B_3$ to a single ``virtual point'' and only focus on the portion of the tour that falls inside $B_3$.

 Similarly to algorithm~A, we can round points inside $B_3$ to grid points (in this case, to a $\delta$-net). Since we no longer have the issue of hops from algorithm~A, the problem that remains to be solved is simply the many-visits TSP described in Lemma~\ref{lemcp}.

\newpage
\paraspace{Algorithm B ($\EPTAS$).}\ \\ {\sc Input:} Set $P$ of $n$ points in $\Dc$, a threshold $\ell$, and a precision parameter $\epsilon > 0$.\\
{\sc Output:} {\sc Yes}, if $P$ admits a tour with scatter at least $\ell$, and {\sc No} if $P$ does not admit a tour with scatter at least $\ell(1 - \epsilon)$.\\

\vspace{-0.1in}
\begin{compactenum}
\item Let $\delta = \epsilon \ell / 2$. 
\item Find a low-degree point $p \in P$, and let $B_2$ resp.\ $B_3$ be the open balls with center $p$ and radius $2 \ell$ resp.\ $3 \ell$. If no such $p$ exists, output {\sc Yes}.

\item Find the $\delta$-net $\Rc$ of $P \cap B_3$ by the algorithm of Lemma~\ref{lem6}, and let $f\colon (P \cap B_3) \rightarrow \Rc$ map points to their nearest point in $\Rc$. 
For each $v\in \Rc$, compute the sets $f^{-1}(v) = \{x \mid f(x) = v\}$.
\item Construct a graph $G'$ as follows:\\ 
- let $V(G') = \Rc \cup \{q\}$.\\
- for all $\{u,v\} \subseteq \Rc$ add edge $\{u,v\}$ to $E(G')$ iff $d(u,v) \geq \ell$.\\
- for all $u \in \Rc$ add edge $\{q,u\}$ to $E(G')$.
\item Let $n_v = | f^{-1}(v) |$ for every $v \in \Rc$, and let $n_q = n - |P \cap B_3|$.
\item Using algorithm $CP$ (Lemma~\ref{lemcp}), find a many-visits tour $Q$ of $G'$ with multiplicities $n_v$ for every $v \in V(G')$; if there is none, output {\sc No}.
\item Transform $Q$ into a tour $T$ of $P$, by replacing multiple occurrences of every point $v \in \Rc$ with the points in $f^{-1}(v)$ in arbitrary order, and by replacing occurrences of $q$ with the points in $P \setminus B_3$ in arbitrary order.
\item Transform $T$ into a tour of $G_P$ by the algorithm of Lemma~\ref{bondy}.
\item Output {\sc Yes}.
\end{compactenum}

\paraspace{Correctness.}
Again, we show two claims: (1) if algorithm~B outputs {\sc Yes}, then there is a tour of $P$ with scatter at least $\ell(1-\epsilon)$, and (2) if there is a tour of $P$ with scatter at least $\ell$, then algorithm~B outputs {\sc Yes}.

(1) If we obtain {\sc Yes} in step~2, then by Lemma~\ref{dirac} there is a tour of $P$ with scatter at least $\ell$. Suppose that we obtain {\sc Yes} in step~9 and thus steps~1--8 execute successfully. Let $G'_P$ be the graph over $P$ with all edges with length at least $\ell (1- \epsilon)$. We wish to show that $G'_P$ is Hamiltonian. Since step~6 succeeds, a many-visits tour of the graph $G'$ constructed in step~4 is found. Since the number of visits in each vertex equals the number of points from $P$ mapped to that vertex, in step~7 we successfully recover a tour $T$ of $P$. 

We wish to show that tour $T$ is a Hamiltonian tour of some graph $G^{*}_P$ that can be obtained from $G'_P$ by repeatedly connecting pairs of vertices whose sum of degrees is at least $n$ (see Lemma~\ref{bondy}). Then, the existence of a Hamiltonian cycle in $G'_P$ follows by Lemma~\ref{bondy} from the existence of a Hamiltonian cycle in $G^{*}_P$. 

In $T$ there are two types of edges: (a) edges between points inside of $B_3$, and (b) edges between a point inside $B_3$ and a point outside of $B_3$. Observe that there are no edges between two points outside of $B_3$, since all points outside of $B_3$ were mapped to the same vertex $q$, and we had no self-edges in $G'$. Consider an arbitrary edge $\{x,y\}$ of $T$. We need to show that $d(x,y) \geq \ell(1-\epsilon)$. Suppose that $\{x,y\}$ is of type (a). By the construction of $G'$ it follows that $d(f(x),f(y)) \geq \ell$, and by Lemma~\ref{lem6}(i) it follows that $d(x,y) \geq \ell - 2\delta = \ell(1-\epsilon)$. Suppose now that $\{x,y\}$ is of type (b), and $x \notin B_3$. If $y \in B_2$, then clearly $d(x,y) \geq \ell$. Otherwise, if $y \notin B_2$, then both $x$ and $y$ are high-degree points (since $B_2$ contains all low-degree points), and thus the sum of their degrees in $G'_P$ is at least $n$, therefore $G'_P$ can safely be augmented with the edge $\{x,y\}$, without affecting its Hamiltonicity.\\

\vspace{-0.1in}
(2) Suppose there exists a tour $T^*$ of $P$ with scatter at least $\ell$ (\ie, a Hamiltonian cycle of $G_P$). 
If there is no low-degree point, step 2 correctly outputs \textsc{Yes}.
Otherwise $B_2$ is centered at a low-degree point and therefore we can assume that $T^*$ has the hop-structure described in Lemma~\ref{lem3} with respect to $B_2$. 
In words, there are no edges of $T^*$ with both endpoints outside of $B_2$. Then, it is easy to check that in steps~3,4 and 5, $f(\cdot)$ maps $T^*$ to a many-visits TSP of the constructed graph $G'$, matching the required multiplicities. Therefore, step~6 succeeds. Step~7 always succeeds, but the resulting tour $T$ may not be the same as $T^*$. In particular, $T$ might have short edges that are not contained in $G_P$, nonetheless, they are contained in an augmentation of $G_P$ as described in part (1), therefore step~8 succeeds. 

\paraspace{Running time.}

The costs of step~6 and step~8 dominate the overall running time. The running time of these steps can be bounded using Lemma~\ref{lemcp}, respectively Lemma~\ref{bondy}, obtaining the total running time $O(n^3 + 2^{O(k \log{k})} + k^3 \cdot \log{n})$, where $k = |\Rc|$. Using Lemma~\ref{lem6}(ii), we obtain $k \leq \left(\frac{6\ell}{\delta} + 1 \right)^d \leq (13/\epsilon)^d$, finishing the proof.

\section{Hardness in high-dimensional Euclidean spaces}\label{sec4}

Trevisan~\cite{Trevisan} showed that the standard TSP is $\APX$-hard in the $\log{n}$-dimensional space $\mathbb{R}^{\log{n}}$, where the distances are computed with respect to an arbitrary $\ell_p$ norm.
The main ingredient of his proof is an embedding of graphs with edge costs in $\{1,2\}$ into $\mathbb{R}^{\log{n}}$.
The embedding is done by assigning binary vectors as labels to the vertices, such that the Hamming distance between two vertex-labels is $D_1$ if the edge between the respective vertices has cost one, and $D_2$ otherwise.
The values of $D_1$ and $D_2$ differ by a constant factor.

The proof relies on the fact that there is a value $B_0$ such that TSP with edge costs $\{1,2\}$ is $\APX$-hard, even if each vertex is incident to at most $B_0$ cost-one edges.

The natural idea to adapt Trevisan's result to MSTSP is to invert the graph: each edge of cost one becomes an edge of cost two and vice versa.
In the new graph, however, the previously used encoding of vertices is not suitable anymore.
Instead of decreasing the Hamming distance of vectors in the case of cost-one edges, now we have to \emph{increase} the Hamming distance in case of cost-two edges.
Our proof uses the general outline of Trevisan's result, but differs in certain details. \\

We first show the hardness of approximating MSTSP in $O(2^n)$-dimensional Hamming spaces. The nature of MSTSP allows us to use a gap-introducing reduction instead of an approximation-preserving reduction. Akiyama et al.~\cite{ANS80} showed that it is $\NP$-hard to decide whether a bipartite $3$-regular graph has a Hamiltonian cycle. Given such a graph, we associate with each vertex a vector from $\{0,1\}^{3\cdot 2^n}$, and consider the edge costs to be the Hamming distance between vertices (\ie, the number of positions where the two vectors are distinct).

For the reduction we use first-order Reed-Muller codes and their complements.
A first-order Reed-Muller code $H_n \subset \{0,1\}^{2^n}$ is a set of $n$ vectors of length $2^n$ such that the pairwise Hamming distance of two code-words is \emph{exactly} $2^n/2$.
The complement $\bar{u}$ of a code-word $u$ is $1^{2^n} - u$, \ie, we invert each entry of $u$.
Note that the Hamming distance of $u$ and $\bar{u}$ is $2^n$, but for all $u'\in H_n$, $u' \neq u$, the Hamming distance of $\bar{u}$ and $u'$ is $2^n/2$.
Let $\bar{H}_n$ be the code that for each vector $u \in H_n$ contains the vector $\bar{u}$. \\

Let $G$ be a bipartite $3$-regular graph with $n$ vertices. It is well-known that such a graph has a $3$-edge coloring that can be found efficiently. Consider such an edge-coloring.
We fix an ordering of the edges of $G$ of the same color, and refer to the $k$th edge of color $i$ for $1 \le k \le n$ and $1 \le i \le 3$.

To each vertex of $G$ we assign a binary vector of length $3 \cdot 2^n$.
The first $2^n$ entries belong to the first color, the second $2^n$ entries to the second color, and the third $2^n$ entries to the third color.

Let $e=\{u,v\}$ be the $k$th edge of color $i$.
Then we set the $i$th $2^n$ entries of $u$ to the $k$th code-word $a$ of $H_n$, and we set the $i$th $2^n$ entries of $v$ to $\bar{a}$.
It does not matter which of the two vertices obtains which of the two code-words.

\begin{lemma}
    Let $u,v$ be two vertices of $G$, and let $d_H$ be the Hamming distance between the vertices (\ie, $3 \cdot 2^n$-vectors).
    \begin{compactenum}
        \item If $\{u,v\} \in E(G)$, then $d_H(u,v) = 2 \cdot 2^n$; and
        \item if $\{u,v\} \notin E(G)$, then $d_H(u,v) = 3\cdot 2^n/2$.
    \end{compactenum}
\end{lemma}
\begin{proof}
    If $e = \{u,v\} \in E(G)$, there is a $k$ and an $i$ such that $e$ is the $k$th edge of color $i$, and therefore all of the $i$th $2^n$ entries of $u$ and $v$ disagree.
    For the two remaining colors, there is no edge between $u$ and $v$, and therefore each of the two remaining pairs of vectors disagree in exactly $2^n/2$ entries:
    All code-words are from $H_n$ or $\bar{H}_n$, and none of the two pairs has complementary code-words. 
    If $\{u,v\} \notin E(G)$, then the Hamming distance between each pair of vectors is exactly $2^n/2$ and thus the overall distance is $3 \cdot 2^n/2$. \qedd
\end{proof}

The lemma implies that it is $\NP$-hard to obtain a $(3/4+\epsilon)$-approximation for MSTSP in $3 \cdot 2^n$-dimensional Hamming metrics.
Suppose otherwise that a polynomial-time algorithm $\mathcal{A}$ can compute such an approximation.
Then, given a hard instance $G$, and the vertex-labeling described above, we 
conclude that $G$ is Hamiltonian if and only if it admits a tour with scatter $2^{n+1}$ (w.r.t.\ the Hamming-distances between vertex-labels). 
If there is no such tour, then the optimal MSTSP solution has value $3\cdot 2^n/2 < (3/4 + \epsilon) \cdot 2^{n+1}$.
We can therefore use $\mathcal{A}$ to decide the Hamiltonicity of $G$ in polynomial time.\\ 

We still have the problem that the dimension $3 \cdot 2^n$ is too large. 
To obtain a suitable embedding, we have to accept a small error.
We say that an embedding $f\colon V \times V \rightarrow \{0,1\}^k$ is $(k,D_1,D_2,\gamma)$-good if for $u,v \in V$,
\begin{compactenum}
    \item if $\{u,v\}$ is of cost $1$, then $D_1-\gamma \le d_H(f(u),f(v)) \le D_1 + \gamma$; and
    \item if $\{u,v\}$ is of cost $2$, then $D_2-\gamma \le d_H(f(u),f(v)) \le D_2 + \gamma$.
\end{compactenum}
\begin{lemma}[Naor and Naor~\cite{NN90}, Trevisan~\cite{Trevisan}]
    For every $\epsilon > 0$ and positive integer $n$ there is a collection
    $C_{n,\epsilon} \subseteq \{0, 1\}^{k(n,\epsilon)}$ such that $|C_{n,\epsilon}| = n$, $k(n,\epsilon) = O((\log n)/\text{poly}(1/\epsilon))$, and for each pair of
two elements $u, v \in C_{n,\epsilon}$ we have $k(n,\epsilon)(1/2 - \epsilon) \le d_H(u, v) \le k(n,\epsilon)(1/2 +\epsilon)$.
Furthermore, there is a procedure that computes $C_{n,\epsilon}$ in $\text{poly}(n, 1/\epsilon)$
time.
\end{lemma}
Observe that inverting a code-word only changes the sign of the error. 
We can therefore apply the previous results with respect to the smaller code-words and obtain the following.
\begin{lemma}
    Let $G$ be a complete $n$-vertex graph with edge costs in $\{1,2\}$ such that the induced subgraph of the cost-two edges is bipartite and $3$-regular.
    There is a polynomial-time algorithm that, for an arbitrary $\gamma > 0$, finds a $(k,3D/2,2D,\gamma)$-good embedding of $V(G)$, where $k= O\left(\log (3 n)/\gamma^2\right)$ and $D = k/2$.
\end{lemma}

\begin{theorem}\label{thm:hardness}
    There is a constant $c$ such that
    it is $\NP$-hard to obtain a $(3/4)^{1/p}$-approximate solution to MSTSP in $\mathbb{R}^d$ for $d \ge c\log n$ and distances according to the $\ell_p$ norm.
\end{theorem}

Theorem~\ref{thm:hardness} follows from the previous discussion and using that for $u,v \in \{0,1\}^n$, $d_p(u,v) = d_H(u,v)^{1/p}$, where $d_p$ is the distance according to $\ell_p$-norm. 
Note that $\{0,1\}^n \subseteq \mathbb{R}^n$ and we can interpret $u$ and $v$ as vectors in the Euclidean space. \\

For large $d$, the running time of algorithm~B is dominated by the term $2^{K \log K}$, where $K = (13/\epsilon)^d$.
The $\APX$-hardness in the $c\log n$-dimensional Euclidean space shows that, except for minor terms, according to our current knowledge, this running time cannot be improved with respect to $d$:
An approximation scheme with running time $2^{2^{o_\epsilon(d)}}$ would imply that we can solve the Hamiltonian cycle problem in bipartite $3$-regular graphs in time $2^{o(n)}$.  
We observe that the reduction from $3$-SAT to Hamiltonian cycle in the target class of graphs increases the parameters only linearly, therefore it preserves subexponential time. (We refer to the hardness proof in~\cite{ANS80}, and the additional details in~\cite{GJT}. We remark, that, since we do not require planarity, the step of removing crossings is not required.) Hence, the existence of such an approximation scheme would contradict the exponential time hypothesis (ETH)~\cite{ipz} and Theorem~\ref{thm:eth} follows.

We leave open the question of whether the dependence on $\epsilon$ can be improved, and whether there exist more efficient approximation schemes specifically tailored to the low-dimensional case.

\paraspace{Acknowledgment.} We would like to thank Karteek Sreenivasaiah for several helpful discussions.

\bigskip

\bibliographystyle{plain}
\bibliography{MSTSP}

\begin{thebibliography}{10}

\bibitem{ANS80}
Takanori Akiyama, Takao Nishizeki, and Nobuji Saito.
\newblock {NP}-completeness of the {H}amiltonian cycle problem for bipartite
  graphs.
\newblock {\em Journal of Information Processing}, 3(2):73--76, 1980.

\bibitem{Arkin}
Esther~M. Arkin, Yi-Jen Chiang, Joseph S.~B. Mitchell, Steven~S. Skiena, and
  Tae-Cheon Yang.
\newblock On the maximum scatter {TSP}.
\newblock In {\em Proceedings of the Eighth Annual ACM-SIAM Symposium on
  Discrete Algorithms}, SODA '97, pages 211--220, Philadelphia, PA, USA, 1997.
  Society for Industrial and Applied Mathematics.

\bibitem{arora}
Sanjeev Arora.
\newblock Polynomial time approximation schemes for euclidean traveling
  salesman and other geometric problems.
\newblock {\em J. {ACM}}, 45(5):753--782, 1998.

\bibitem{sunil-arya}
Sunil Arya and David~M. Mount.
\newblock A fast and simple algorithm for computing approximate euclidean
  minimum spanning trees.
\newblock In {\em Proceedings of the Twenty-Seventh Annual ACM-SIAM Symposium
  on Discrete Algorithms}, SODA '16, pages 1220--1233, Philadelphia, PA, USA,
  2016. Society for Industrial and Applied Mathematics.

\bibitem{Bartal}
Yair Bartal, Lee-Ad Gottlieb, and Robert Krauthgamer.
\newblock The traveling salesman problem: Low-dimensionality implies a
  polynomial time approximation scheme.
\newblock In {\em Proceedings of the Forty-fourth Annual ACM Symposium on
  Theory of Computing}, STOC '12, pages 663--672, New York, NY, USA, 2012. ACM.

\bibitem{barvinok}
Alexander~I. Barvinok.
\newblock Two algorithmic results for the traveling salesman problem.
\newblock {\em Math. Oper. Res.}, 21(1):65--84, 1996.

\bibitem{BFJ}
Alexander~I. Barvinok, S{\'a}ndor~P. Fekete, David~S. Johnson, Arie Tamir,
  Gerhard~J. Woeginger, and Russell Wodroofe.
\newblock The geometric maximum traveling salesman problem.
\newblock {\em J. ACM}, 50:641--664, 2003.

\bibitem{bollobas}
B{\'e}la Bollob{\'a}s.
\newblock {\em Modern Graph Theory}.
\newblock Springer, 1998.

\bibitem{BondyChvatal}
J.~A. Bondy and V.~Chvatal.
\newblock A method in graph theory.
\newblock {\em Discrete Math.}, 15(2):111--135, January 1976.

\bibitem{Chr76}
Nicos Christofides.
\newblock Worst-case analysis of a new heuristic for the travelling salesman
  problem.
\newblock Technical Report 388, Carnegie Mellon University, 1976.

\bibitem{Clarkson1999}
L.~K. Clarkson.
\newblock Nearest neighbor queries in metric spaces.
\newblock {\em Discrete {\&} Computational Geometry}, 22(1):63--93, 1999.

\bibitem{Cosmadakis}
Stavros~S. Cosmadakis and Christos~H. Papadimitriou.
\newblock The traveling salesman problem with many visits to few cities.
\newblock {\em {SIAM} J. Comput.}, 13(1):99--108, 1984.

\bibitem{TOPP}
E.~D. Demaine, J.~S.~B. Mitchell, and J.~O'Rourke.
\newblock The open problems project, 2012.

\bibitem{Doroshko}
N.~N. Doroshko and V.~I. Sarvanov.
\newblock The minimax traveling salesman problem and hamiltonian cycles in
  powers of graphs ({Russian}).
\newblock {\em Vestsi Akad. Navuk BSSR, Ser. Fiz.-Mat. Navuk}, (6):119--120,
  1981.

\bibitem{Fekete}
S\'{a}ndor~P. Fekete.
\newblock Simplicity and hardness of the maximum traveling salesman problem
  under geometric distances.
\newblock In {\em SODA, 1999}, pages 337--345, Philadelphia, PA, USA, 1999.
  Society for Industrial and Applied Mathematics.

\bibitem{Friggstad}
Zachary Friggstad, Mohsen Rezapour, and Mohammad~R. Salavatipour.
\newblock Local search yields a {PTAS} for k-means in doubling metrics.
\newblock {\em CoRR}, abs/1603.08976, 2016.

\bibitem{GJT}
M.~R. Garey, D.~S. Johnson, and R.~Endre Tarjan.
\newblock The planar hamiltonian circuit problem is np-complete.
\newblock {\em SIAM Journal on Computing}, 5(4):704--714, 1976.

\bibitem{Gupta}
Anupam Gupta, Robert Krauthgamer, and James~R. Lee.
\newblock Bounded geometries, fractals, and low-distortion embeddings.
\newblock In {\em Proceedings of the 44th Annual IEEE Symposium on Foundations
  of Computer Science}, FOCS '03, pages 534--, Washington, DC, USA, 2003. IEEE
  Computer Society.

\bibitem{Gutin}
Gregory Gutin, Abraham Punnen, Alexander Barvinok, Edward~Kh. Gimadi, and
  Anatoliy~I. Serdyukov.
\newblock {\em The Traveling Salesman Problem and Its Variations}.
\newblock Springer, 2002.

\bibitem{sariel_book}
Sariel Har-Peled.
\newblock {\em Geometric Approximation Algorithms}.
\newblock American Mathematical Society, Boston, MA, USA, 2011.

\bibitem{ipz}
Russell Impagliazzo, Ramamohan Paturi, and Francis Zane.
\newblock Which problems have strongly exponential complexity?
\newblock {\em Journal of Computer and System Sciences}, 63(4):512 -- 530,
  2001.

\bibitem{Karpinski}
Marek Karpinski, Michael Lampis, and Richard Schmied.
\newblock New inapproximability bounds for {TSP}.
\newblock {\em J. Comput. Syst. Sci.}, 81(8):1665--1677, 2015.

\bibitem{Kowalik}
Lukasz Kowalik and Marcin Mucha.
\newblock Deterministic $7/8$-approximation for the metric maximum {TSP}.
\newblock {\em Theoretical Computer Science}, 410(47–49):5000 -- 5009, 2009.

\bibitem{Larusic}
John Larusic, Abraham~P. Punnen, and Eric Aubanel.
\newblock Experimental analysis of heuristics for the bottleneck traveling
  salesman problem.
\newblock {\em Journal of Heuristics}, 18(3):473--503, June 2012.

\bibitem{Lawler}
E.L. Lawler, D.B. Shmoys, A.H.G.R. Kan, and J.K. Lenstra.
\newblock {\em The Traveling Salesman Problem}.
\newblock John Wiley \& Sons, 1985.

\bibitem{Mitchell96}
Joseph S.~B. Mitchell.
\newblock Guillotine subdivisions approximate polygonal subdivisions: A simple
  polynomial-time approximation scheme for geometric {TSP}, k-{MST}, and
  related problems.
\newblock {\em SIAM J. Comput}, 28:402--408, 1996.

\bibitem{Saurabh}
Nabil~H. Mustafa and Saurabh Ray.
\newblock {PTAS} for geometric hitting set problems via local search.
\newblock In {\em Proceedings of the 25th {ACM} Symposium on Computational
  Geometry, Aarhus, Denmark, June 8-10, 2009}, pages 17--22, 2009.

\bibitem{NN90}
Joseph Naor and Moni Naor.
\newblock Small-bias probability spaces: Efficient constructions and
  applications.
\newblock In {\em {STOC}}, pages 213--223. {ACM}, 1990.

\bibitem{Yannakakis}
Christos~H. Papadimitriou and Mihalis Yannakakis.
\newblock The traveling salesman problem with distances one and two.
\newblock {\em Math. Oper. Res.}, 18(1):1--11, February 1993.

\bibitem{Parker}
R.Gary Parker and Ronald~L Rardin.
\newblock Guaranteed performance heuristics for the bottleneck travelling
  salesman problem.
\newblock {\em Oper. Res. Lett.}, 2(6):269--272, March 1984.

\bibitem{Handbook}
J.-R. Sack and J.~Urrutia, editors.
\newblock {\em Handbook of Computational Geometry}.
\newblock North-Holland Publishing Co., 2000.

\bibitem{Sarvanov}
V.~I. Sarvanov.
\newblock A minimax traveling salesman problem on a plane: complexity of an
  approximate solution. ({Russian}).
\newblock {\em Dokl. Akad. Nauk Belarusi}, 39(6):16--19, 1995.

\bibitem{Talwar}
Kunal Talwar.
\newblock Bypassing the embedding: Algorithms for low dimensional metrics.
\newblock In {\em Proceedings of the Thirty-sixth Annual ACM Symposium on
  Theory of Computing}, STOC '04, pages 281--290, New York, NY, USA, 2004. ACM.

\bibitem{Trevisan}
Luca Trevisan.
\newblock When {H}amming meets {E}uclid: The approximability of geometric {TSP}
  and {S}teiner tree.
\newblock {\em SIAM J. Comput.}, 30(2):475--485, April 2000.

\end{thebibliography}

\end{document}